\documentclass[journal]{IEEEtran}
\usepackage{cite}
\usepackage{amsmath,amssymb,amsfonts}
\usepackage{algorithmic}
\usepackage{graphicx}
\usepackage{textcomp}
\usepackage{graphicx}
\usepackage{amsmath}
\usepackage{amssymb}
\usepackage{amsthm}
\usepackage{enumerate}
\usepackage{xcolor}
\usepackage{amstext}
\usepackage{tikz}
\usepackage{algorithmic}
\usepackage{algorithm}
\usepackage{siunitx}
\usepackage{balance}
\usepackage{cuted}
\usepackage{array}
\usepackage{float}
\usepackage{multirow}
\usepackage{psfrag}

\theoremstyle{plain}
\newtheorem{thm}{Theorem}
\newtheorem{lem}{Lemma}
\newtheorem{prop}{Proposition}

\theoremstyle{definition}
\newtheorem{defn}{Definition}

\newtheorem{prob}{Problem}

\theoremstyle{remark}
\newtheorem{rem}{Remark}

\newcommand{\argmin}{\mathrm{argmin}}

\newcommand{\R}{\mathbb{R}}

\newcommand{\one}{\mathbf{1}}
\newcommand{\wm}{W_\textrm{M}}

\newcommand{\wadmm}{W_\textrm{admm}}
\newcommand{\wtadmm}{\hat{W}_\textrm{admm}}
\newcommand{\wbadmm}{\bar{W}_\textrm{admm}}

\DeclareMathOperator{\tr}{Tr}
\newcolumntype{L}[1]{>{\centering\let\newline\\\arraybackslash\hspace{0pt}}m{#1}}

\ifCLASSINFOpdf
\else
\fi
\begin{document}
%
\title{Distributed computation of fast consensus weights using ADMM}
%
%
%

\author{Kiran Rokade and Rachel Kalpana Kalaimani
\thanks{This work has been partially supported by DST-INSPIRE Faculty Grant, Department of Science and Technology (DST), Govt. of India (ELE/16-17/333/DSTX/RACH).}
\thanks{Kiran Rokade is with the Department of Electrical and Computer Engineering, Cornell University and Rachel Kalpana Kalaimani is with the Department of Electrical Engineering, Indian Institute of Technology Madras (e-mail: kvr36@cornell.edu, rachel@ee.iitm.ac.in)}}

\maketitle

\begin{abstract}
We consider the problem of achieving average consensus among multiple agents, where the inter-agent communication network is depicted by a graph. We consider the discrete-time consensus protocol where each agent updates its value as a weighted average of its own value and those of its neighbours. Given a graph, it is known that there exists a set of `optimal weights' such that the agents reach average consensus asymptotically with an optimal rate of convergence. However, existing methods require the knowledge of the entire graph to compute these optimal weights. We propose a method for each agent to compute its set of optimal weights locally, i.e., each agent only has to know who are its neighbours. The method is derived by solving a matrix norm minimization problem subject to linear constraints in a distributed manner using the Alternating Direction Method of Multipliers (ADMM). We illustrate our results using numerical examples and compare our method with an existing method called the Metropolis weights, which are also computed locally. 
\end{abstract}


%
\IEEEpeerreviewmaketitle

\section{Introduction}

In a consensus problem, a group of agents seeks to agree upon a certain quantity of interest. In many applications, this quantity of interest is the average of the initial states of the agents \cite{1440896,CYBENKO1989279}. Consensus problems arise in various settings such as multi-agent formation control \cite{1333200}, estimating a parameter using a network of sensors \cite{1440896}, workload balancing for distributed computation \cite{CYBENKO1989279}, etc. 
These problems have been widely studied in control theory \cite{1239709,1333204,XIAO200465,1431045}.

In all consensus protocols, an important parameter is the time required to reach average consensus. One line of literature focuses on algorithms which achieve consensus in finite time \cite{4282726,5404832,5160026}. The other line of literature looks at algorithms which achieve consensus asymptotically, while trying to improve the rate of this asymptotic convergence \cite{XIAO200465,5707021}. For continuous-time consensus protocols, it is known that the rate of convergence depends on the second-smallest eigenvalue of the graph Laplacian \cite{1333204}, while for weighted discrete-time protocols, it depends on the second-largest absolute value of all eigenvalues of the weight matrix \cite{XIAO200465}. It is then natural to optimize this rate of convergence by appropriately choosing the edge and node weights used by the agents in their update laws for reaching consensus. 

The problem of optimizing the rate of convergence of distributed linear discrete-time consensus protocols has been widely studied in literature and is sometimes referred to as the fastest distributed linear averaging (FDLA) problem \cite{XIAO200465}. In \cite{XIAO200465}, the authors consider the general weighted discrete-time consensus protocol, with possibly asymmetric edge weights. They propose two metrics for quantifying the rate of convergence: the asymptotic convergence factor and the per-step convergence factor. It is shown that the problem of optimizing the per-step convergence factor is a convex optimization problem and hence can be solved efficiently. In \cite{5707021}, the authors consider consensus protocols with symmetric edge weights. Considering the edge and node weights as variables, they try to minimize the distance between the largest and the second-smallest eigenvalue of the weighted graph Laplacian. This in-turn optimizes the rate of convergence of the consensus protocol. In \cite{1576862}, the authors consider the problem of maximizing the second-smallest eigenvalue of the graph Laplacian where the edge weights are functions of the agent values. They propose an iterative method which gives optimal trajectories of agent values by solving an optimization problem at each iteration. Another class of problems which falls under the FDLA framework is that of fastest mixing Markov chains \cite{doi:10.1137/S0036144503423264}. Here, the equation governing the evolution of the probability distribution is similar to the discrete-time consensus protocol. The objective is to find a transition probability matrix, analogous to the weight matrix in a consensus protocol, such that an initial probability distribution converges to a uniform distribution as fast as possible.

While the above methods give weights that achieve the optimal rate of convergence in a consensus protocol, finding these weights require the knowledge of the entire network topology. In other words, there has to be a central entity which has the knowledge of the entire network, can calculate the optimal weights and then broadcast them to the agents. Firstly, there may not exist such a central entity due to constraints on the communication resources or for security reasons. Secondly, even if it does exist, this process of computing and broadcasting the weights has to be repeated every time the network topology changes due to say node failure, edge failure or if new agents enter the network. This can be very inefficient if the network is large. In the spirit of distributed computation, we would want the agents to locally compute these weights which they use in their update laws. An example of locally computed weights is the Metropolis weights \cite{1440896}. To compute them, an agent has to know only the degree of itself and its neighbours. While these weights do achieve consensus, the rate of convergence of the consensus protocol with these weights is not optimal. We seek to have the best of both worlds: locally computed weights whose rate of convergence is optimal. 

A solution to the problem of locally computing the optimal consensus weights has been attempted in \cite{6887328}. There, the authors use a distributed gradient-descent approach on the Schatten $p$-norm of the weight matrix, which is an approximation of the second-largest eigenvalue of the weight matrix. Computing the gradient requires the agents to communicate with its neighbours which are up to $p/2$ hops away. There is a trade-off between locality and optimality: smaller values of $p$ give sub-optimal solutions. We propose an algorithm where the agents need to communicate only with their immediate neighbours to compute the consensus weights. Moreover, the rate of convergence of the consensus protocol using the weights obtained by our algorithm can be arbitrarily close to the optimal value. The algorithm is derived using the Alternating Direction Method of Multipliers (ADMM) \cite{10.1561/2200000016,Eckstein2012}. We summarize our contributions below.

\subsection{Contributions}

\begin{enumerate}
    \item We consider agents which communicate with each other over a network depicted by an undirected graph. We propose an iterative algorithm for the agents to locally compute the weights required for the weighted-average discrete-time consensus protocol. We show that the rate of convergence of the consensus protocol using these weights converges to the optimal value. The algorithm is derived by solving a matrix norm minimization problem in a distributed manner using ADMM. The result is given in Section \ref{sec:result}.
    \item For some applications, we propose a variation of our algorithm, where, at every iteration, the agents compute the weights and update their values using these weights. Using a numerical example, we show that this variation of our algorithm performs better in terms of the rate of convergence of the consensus protocol than the locally computed Metropolis weights. This is done in Section \ref{subsec:admm_live_example}.
\end{enumerate}

The rest of the paper is organized as follows. We next summarize some basic notations which will be used throughout the paper. In Section \ref{sec:prob_form}, we mathematically formulate the problem of computing the optimal consensus weights. In Section \ref{sec:result}, we propose our algorithm to locally compute these optimal weights. We also state our main result which proves the convergence of this algorithm. The proof of the result is deferred to the appendix. In Section \ref{sec:examples}, we present some numerical examples to illustrate the convergence of our algorithm and to compare its performance with the Metropolis weights. Finally, we give some concluding remarks in Section \ref{sec:conclusion}.

\subsection{Notation}

All vectors are of length $n$ and all matrices are of size $n \times n$. $\one$ is the vector $\begin{bmatrix} 1 &\dots &1 \end{bmatrix}^T$ of all ones. $e_i$ is the $i^\text{th}$ standard basis vector with $1$ as the $i^\text{th}$ entry and zeros elsewhere. For a vector $x$: $x_i$ or $(x)_i$ denotes its $i^\text{th}$ element, $||x||$ denotes its 2-norm. For a matrix $A$: $\rho(A)$ denotes its spectral radius or the maximum absolute eigenvalue, $||A||$ denotes its induced 2-norm or the maximum singular value, $||A||_F$ denotes its Frobenius norm, $A_{ij}$ or $(A)_{ij}$ denotes its $(i,j)^{\text{th}}$ element, $\tr(A)$ denotes its trace, $A > 0$ denotes $A_{ij} > 0$ for all $(i,j)$. For a set $B$, $|B|$ denotes its cardinality. For a function $f: \R^n \rightarrow \R$, $(\partial/\partial x)f = \begin{bmatrix} (\partial/\partial x_1)f &\dots &(\partial/\partial x_n)f \end{bmatrix}^T$ denotes the gradient of $f$ with respect to the vector $x \in \R^n$. For a function $f: \R^{n \times n} \rightarrow \R$, $$(\partial/\partial X)f = \begin{bmatrix} (\partial/\partial X_{11})f &\dots &(\partial/\partial X_{1n})f \\ \vdots &\ddots &\vdots \\ (\partial/\partial X_{n1})f &\dots &(\partial/\partial X_{nn})f \end{bmatrix}$$ denotes the gradient of $f$ with respect to the matrix $X \in \R^{n \times n}$.

\section{Problem Formulation}
\label{sec:prob_form}

\subsection{The average consensus problem}

Consider a set $V = \{1, \dots, n\}$ of $n$ agents, each having a scalar initial value $x_i(0) \in \R, i \in \{1, \dots, n\}$. Assume that the inter-agent communication is governed by an undirected, connected graph $G = (V,E)$, where $E$ is the set of all undirected edges of $G$. In other words, agents $i$ and $j$ can exchange values with each other if and only if $(i,j) \in E$. Note that an agent always knows its own value, hence, for the ease of notation, we assume that $(i,i) \in E$ for all $i \in \{1, \dots, n\}$. Let $N_i = \{j \in V: (i,j) \in E\}$ be the set of all neighbours of agent $i$. 

We want the agents to reach average consensus, i.e., each agent must compute the average $x_\textrm{avg}(0) = (1/n) \sum_{i = 1}^n x_i(0)$ of the initial values in a distributed manner by communicating only with its neighbours. We consider the distributed linear iterative protocol where each agent $i$ updates its value as
\begin{align}
    \label{eq:algo_elementwise}
    x_i(t+1) = W_{ii} x_i(t) + \sum_{j \in N_i} W_{ij} x_j(t), \ i \in \{1, \dots, n\},
\end{align}
where $t \in \{0, 1, 2, \dots$\} is the time and $W_{ij}$ is the weight assigned by agent $i$ to agent $j$'s value. Due to the communication constraint imposed by the graph, we fix $W_{ij} = 0$ if $(i,j) \notin E$. 
\begin{rem}
    \label{rem:asym_W}
    Note that since the graph $G$ is undirected, $(i,j) \in E$ implies $W_{ij}$ and $W_{ji}$ can both be nonzero. However, they can take different values in general. 
\end{rem}
Now, the protocol \eqref{eq:algo_elementwise} can be written as
\begin{align}
    \label{eq:algo_vector}
    x(t+1) = W x(t),
\end{align}
where $W \in \R^{n \times n}$ is called the \emph{weight matrix} and $x(t) = \begin{bmatrix} x_1(t) &\dots &x_n(t) \end{bmatrix}^T \in \R^n$ is the vector of agent values at time $t$. We refer to \eqref{eq:algo_vector} as the \emph{consensus protocol}. Let $\bar{x} = x_\textrm{avg}(0) \one$ be the vector with $x_\textrm{avg}(0)$ as all its entries. From \eqref{eq:algo_vector}, we have $x(t) = W^t x(0)$. Thus, the agents reach average consensus, i.e., $\lim_{t \rightarrow \infty} x(t) = \bar{x}$ if and only if 
\begin{align}
    \label{eq:W_converges_one_one^T/n}
    \lim_{t \rightarrow \infty} W^t = \one \one^T/n.
\end{align}
Following result from \cite{XIAO200465} gives a necessary and sufficient condition for \eqref{eq:W_converges_one_one^T/n} to be true.
\begin{prop}{\cite{XIAO200465}}
    \label{prop:boyd_W_condition}
    The consensus protocol \eqref{eq:algo_vector} reaches average consensus, i.e., $\lim_{t \rightarrow \infty} x(t) = \bar{x}$ if and only if 
    \begin{align}
    \label{eq:W_consensus_condition}
        W \one = \one, \ W^T \one = \one, \ \rho(W - \one \one^T/n) < 1.
    \end{align}
\end{prop}
Note that the weights $W_{ij}$ are allowed to be negative in general. Condition \eqref{eq:W_consensus_condition} ensures that in the consensus protocol given by \eqref{eq:algo_vector}:
\begin{itemize}
    \item the average $\one^T x(t)/n$ of the agent values is invariant across time $t$,
    \item any vector in the `agreement space', i.e., in the span of $\one$, is invariant with respect to $W$,
    \item all eigenvalues of $W$ other than the eigenvalue $1$ are strictly inside the unit circle.
\end{itemize}
Given that $W$ satisfies \eqref{eq:W_consensus_condition}, Proposition \ref{prop:boyd_W_condition} guarantees $\lim_{t \rightarrow \infty} x(t) = \bar{x}$. We are interested in maximizing the rate at which this convergence occurs. 

\subsection{Optimal weights}

Define $e(t) = x(t) - \bar{x}$ as the consensus error. Then, from the consensus protocol \eqref{eq:algo_vector}, we can write the dynamics of the consensus error as
\begin{align}
    \label{eq:error_dynamics}
    e(t+1) = (W - \one \one^T/n) e(t).
\end{align}
Now, one way to characterize the rate of convergence of the consensus protocol is by the spectral radius $\rho(W - \one \one^T/n)$. This quantity is the largest absolute eigenvalue of $W$ other than the eigenvalue $1$. It is called the asymptotic convergence factor \cite{XIAO200465}. The weight matrix $W$ which gives the best asymptotic convergence factor can be found by solving the optimization problem
\begin{align}
    \label{eq:W_optimization_rho}
    & \mathrm{minimize} & & \rho(W - \one \one^T/n) \tag{P1} \\
    & \mathrm{subject \ to} & & W \one = \one, \ W^T \one = \one, \nonumber \\
    & & & W_{ij} = 0, (i,j) \notin E, \nonumber
\end{align}
where the constraint 
\begin{align}
    \label{eq:topological_constraint}
    W_{ij} = 0, (i,j) \notin E
\end{align}
is the topological constraint imposed by the graph $G$. It is known that $\rho(W - \one \one^T/n)$ is in general a non-convex function of $W$ and hence \eqref{eq:W_optimization_rho} is a hard problem to solve \cite{doi:10.1137/0609040}. 

We replace the objective function $\rho(W - \one \one^T/n)$ in \eqref{eq:W_optimization_rho} by the convex function $||W - \one \one^T/n||$. By definition of the induced matrix norm, we have
\begin{align}
    \label{eq:per-step_factor}
    ||W - \one \one^T/n|| = \sup_{e(t) \neq 0, t \geq 0} \frac{||e(t+1)||}{||e(t)||},
\end{align}
i.e., $||W - \one \one^T/n||$ captures the worst case one-step increase in the consensus error. Hence, $||W - \one \one^T/n||$ is known as the per-step convergence factor \cite{XIAO200465}. Now, consider the convex problem
\begin{align}
    \label{eq:W_optimization_2-norm}
    & \mathrm{minimize} & & ||W - \one \one^T/n|| \tag{P2} \\
    & \mathrm{subject \ to} & & W \one = \one, \ W^T \one = \one, \nonumber \\
    & & & W_{ij} = 0, (i,j) \notin E. \nonumber
\end{align}
We denote a solution of \eqref{eq:W_optimization_2-norm} by $W^*$. Note that since $W \neq W^T$ in general (refer Remark \ref{rem:asym_W}), we have $\rho(W - \one \one^T/n) \leq ||W - \one \one^T/n||$. We show that if the graph $G$ is connected, then $W^*$ satisfies $\rho(W^* - \one \one^T/n) < 1$.
\begin{lem}
    \label{lem:convex_relaxation}
    The weight matrix $W^*$ obtained by solving the convex problem \eqref{eq:W_optimization_2-norm} satisfies the average consensus condition given in \eqref{eq:W_consensus_condition}.
\end{lem}
The proof of Lemma \ref{lem:convex_relaxation} is given in Appendix \ref{sec:appendix_proof_conv_relax_lemma}. Henceforth, we refer to $||W - \one \one^T/n||$ as simply the \emph{convergence factor} of $W$. Thus, $W^*$ is a weight matrix which achieves consensus and gives an optimal convergence factor for the consensus protocol. Due to the topological constraint $W_{ij} = 0, (i,j) \notin E$, solving \eqref{eq:W_optimization_2-norm} requires the knowledge of the entire graph $G$. Thus, \eqref{eq:W_optimization_2-norm}, in its original form, can only be solved in a `centralized manner'. We later propose to solve the problem in a distributed manner using ADMM. 

\subsection{Locally calculated weights}

Distributed algorithms require that each agent computes the average $x_\textrm{avg}(0)$ using only local information from its neighbours. In the same spirit, we would want each agent $i$ to calculate its optimal weights $\{{W^*}_{ij}, j = 1, \dots, n\}$ locally, without the knowledge of the entire graph. Before looking at the problem of locally computing the optimal weights, we look at a set of weights, which, although not optimal, can be computed locally. 

Consider the set of weights called the local degree weights or the Metropolis weights. We denote the matrix of these weights by $\wm$. For each agent $i \in \{1, \dots, n\}$, let
\begin{align}
    \label{eq:W_metro}
    (\wm)_{ij} = \begin{cases}
            0, & (i,j) \notin E, \\
            \min \big\{\frac{1}{1+|N_i|}, \frac{1}{1+|N_j|}\big\}, & (i,j) \in E, i \neq j, \\
            1 - \sum_{j \in N_i} (\wm)_{ij}, & i = j.
        \end{cases}
\end{align}
Intuitively, with these weights used in the consensus protocol, the information of an agent carries more weight if it has few neighbours. This should speed up the propagation of the information of the not-so-well-connected agents in the network, which will improve the overall rate of convergence of the entire network. These weights are widely used since they are simple to calculate and can be shown to achieve average consensus \cite{1440896}. However, as we shall see later through a numerical example, the convergence factor of $\wm$ can be significantly poor than $W^*$. Our aim is to give a method for computing weights locally, whose convergence factor is optimal. 
\begin{prob}
    \label{prob:main}
    Given an undirected graph $G = (V,E)$, derive a method such that each agent $i \in \{1, \dots, n\}$ can determine its set of weights $\{W_{ij}, j = 1, \dots, n\}$ knowing only its neighbour set $N_i$ such that the weight matrix $W$
    \begin{enumerate}
        \item satisfies the average consensus condition given in \eqref{eq:W_consensus_condition} and
        \item has an optimal convergence factor, i.e., $$||W - \one \one^T/n|| = ||W^* - \one \one^T/n||,$$
        where $W^*$ is a solution of \eqref{eq:W_optimization_2-norm}.
    \end{enumerate}
\end{prob}
We propose a solution to Problem \ref{prob:main} by solving \eqref{eq:W_optimization_2-norm} in a distributed manner. We do this using the well-known Alternating Direction Method of Multipliers (ADMM) \cite{Eckstein2012}. It gives fairly accurate results in relatively fewer number of iterations than other methods \cite{10.1561/2200000016}. Due to this, it has been widely used in solving practical optimization problems, e.g. \cite{6584818,7170787}. Our result is presented in the next section.

\section{Main Result}
\label{sec:result}

In this section, we propose our result which shows how \eqref{eq:W_optimization_2-norm} can be solved in a distributed manner over an undirected, connected graph. First, consider the problem
\begin{align}
    \label{eq:W_i_optimization}
    & \mathrm{minimize} & & \sum_{i=1}^n\frac{||W_i - \one \one^T/n||}{n} \tag{P3} \\
    & \mathrm{subject \ to} & & W_i = W_j, (i,j) \in E, \nonumber \\
    & & & W_i \one = \one, \ W_i^T \one = \one, \ i \in \{1, \dots, n\}, \nonumber \\
    & & & (W_i)_{ij} = 0, (i,j) \notin E, i \in \{1, \dots, n\}. \nonumber
\end{align}
Here, we have replaced the $W$ in \eqref{eq:W_optimization_2-norm} by $n$ matrices $W_1, \dots, W_n$, one for each agent in the network. The matrix $W_i$ can be seen as agent $i$'s estimate of the centralized optimal solution $W^*$. Since $G$ is connected, the constraint $W_i = W_j, (i,j) \in E$ implies all $W_i$'s are equal. This implies \eqref{eq:W_i_optimization} is equivalent to \eqref{eq:W_optimization_2-norm}. Introducing these new matrices will enable solving \eqref{eq:W_i_optimization} in a distributed manner. A key step towards this is to impose the topological constraint $(W_i)_{ij} = 0, (i,j) \notin E$ only on the $i^\text{th}$ row of the matrix $W_i$, for each $i \in \{1, \dots, n\}$. This means each agent can impose this constraint on its estimate $W_i$ knowing only its neighbour set $N_i$. We show that \eqref{eq:W_i_optimization} can be solved in a distributed manner. The steps are given in Algorithm \ref{alg:admm_parallel}. The result is formally stated below, its proof can be found in Appendix \ref{sec:appendix_proof_main}. 

\begin{thm}
\label{thm:main}
    Consider the $n$ sequences of matrices $\{W_i(k)\}_{k \geq 0}, i \in \{1, \dots, n\}$, where each sequence is generated by running Algorithm \ref{alg:admm_parallel} in parallel on each agent $i \in \{1, \dots, n\}$. Then, in the limit as $k$ goes to infinity, all $n$ matrices $W_i(k), i \in \{1, \dots, n\}$
    \begin{enumerate}
        \item are equal, i.e.,
        \begin{align*}
            \lim_{k \rightarrow \infty} ||W_i(k) - W_j(k)|| = 0
        \end{align*}
        for all $i,j \in \{1, \dots, n\}$,
        \item satisfy the constraints of problem \eqref{eq:W_optimization_2-norm} and
        \item \label{thm_point:obj_convergence} give an optimal convergence factor for the consensus protocol \eqref{eq:algo_vector}, i.e., 
        \begin{align*}
            \lim_{k \rightarrow \infty}|| W_i(k) - \one \one^T/n|| = ||W^* - \one \one^T/n||
        \end{align*}
        for all $i \in \{1, \dots, n\}$, where $W^*$ is a solution of the centralized problem \eqref{eq:W_optimization_2-norm}.
    \end{enumerate}
\end{thm}

\begin{algorithm}[ht]
\caption{Calculating locally an estimate of the optimal weight matrix $W^*$ at agent $i$ of the network $G$}
\label{alg:admm_parallel}
\begin{algorithmic}
	\STATE \textbf{Initialize:} $\rho > 0, W_i(0) = 0_{n \times n}, a_i(0) = 0_{n \times 1}, b_i(0) = 0_{n \times 1}, M_i(0) = 0_{n \times n}$
	\STATE Send $W_i(0)$ to all neighbours $j \in N_i$.
	\STATE Receive $W_j(0)$ from all neighbours $j \in N_i$.
	\STATE \textbf{Iterate:} \FOR{$k \geq 0$}
	\STATE \textbf{Primal update:}
	\STATE Evaluate $W_i(k+1)$ as per \eqref{eq:wi_final}.
	\STATE \textbf{Exchange values:}
	\STATE Send $W_i(k+1)$ to all neighbours $j \in N_i$.
	\STATE Receive $W_j(k+1)$ from all neighbours $j \in N_i$.
	\STATE \textbf{Dual update:} \begin{align*}
	    a_i(k+1) &= a_i(k) + \rho(W_i(k+1) \one - \one) \\
	    b_i(k+1) &= b_i(k) + \rho(W_i(k+1)^T \one - \one) \\
	    M_i(k+1) &= M_i(k) + \frac{\rho}{2} \sum_{j \in N_i}(W_i(k+1) - W_j(k+1))
	\end{align*}
	\ENDFOR
\end{algorithmic}
\end{algorithm}
It is known that for a convex optimization problem having only equality constraints, the ADMM iterates converge to the optimal objective value and satisfy the constraints of the optimization problem in the limit as the number of iterations goes to infinity \cite[Section 3.2.1]{10.1561/2200000016}. Hence,
to prove Theorem \ref{thm:main}, all we need to show is that Algorithm \ref{alg:admm_parallel} is an ADMM implementation of problem \eqref{eq:W_i_optimization}. This is done in Appendix \ref{sec:appendix_proof_main}. We now make some remarks on the algorithm.

\begin{rem}{(Convergence of the primal variable sequences)}
    ADMM does not guarantee convergence of the primal variables \cite[Section 3.2.1]{10.1561/2200000016}. In the context of Theorem \ref{thm:main}, this implies that, in general, $\lim_{k \rightarrow \infty} W_i(k)$ may not exist. However, this does not affect the practical application of our result since we are only interested in a weight matrix $W$ which satisfies the constraints of \eqref{eq:W_optimization_2-norm} and gives an optimal convergence factor.
\end{rem}

\begin{rem}{(Stopping criterion for Algorithm \ref{alg:admm_parallel})}
    \label{rem:stop_crit}
    A stopping criterion for Algorithm \ref{alg:admm_parallel} is derived in Appendix \ref{sec:appendix_stop_crit}. The criterion can be verified locally as follows. Each agent fixes an arbitrary, small number $\epsilon > 0$. At each iteration $k$, the agent computes its `residual' $R_i(k)$, which captures how far is the agent's estimate $W_i(k)$ from the optimal objective function value (see \eqref{eq:max_residual} further ahead for the definition of $R_i(k)$). The agent is trying to minimize this residual. For a given iteration $k$, we say the stopping criterion of agent $i$ is satisfied if $R_i(k) \leq \epsilon$. Note that for the smallest $k$ for which the stopping criterion of agent $i$ is satisfied, the agent can stop updating its variables and send the fixed value $W_i = W_i(k)$ to its neighbours as long as $R_i(k) \leq \epsilon$ holds. We say that the stopping criterion for Algorithm \ref{alg:admm_parallel} is satisfied if $R_i(k) \leq \epsilon$ for all $i \in \{1, \dots, n\}$. 
\end{rem}

\begin{rem}{(Running Algorithm \ref{alg:admm_parallel} in a distributed, parallel manner)}
    \label{rem:alg_distributed_manner}
    We can verify that Algorithm \ref{alg:admm_parallel} can indeed be executed in a distributed manner. Each agent $i$ maintains a primal variable $W_i(k)$ and dual variables $a_i(k)$, $b_i(k)$ and $M_i(k)$ at each iteration $k$. The penalty parameter $\rho > 0$ can be arbitrarily fixed by the agent, although a standard choice is $\rho = 1$. At each iteration, each agent updates its primal variable, exchanges the updated primal variables with its neighbours and updates its dual variables. Thus, Algorithm \ref{alg:admm_parallel} can run in a distributed manner, where each agent has to exchange values only with its immediate neighbours.
    
    \noindent From the primal update equation \eqref{eq:wi_final} of the algorithm, it is clear that to calculate $W_i(k+1)$, agent $i$ requires the $W_j(k)$ values only from the previous iteration of its neighbours. Hence, it does not have to wait for other agents to finish their updates before performing its update. Thus, Algorithm \ref{alg:admm_parallel} can run in parallel across all agents.
\end{rem}

\begin{rem}{(Information required in running Algorithm \ref{alg:admm_parallel})}
    \label{rem:admm_vs_metro_info_&_comm}
    To run Algorithm \ref{alg:admm_parallel} on all agents, the agents have to initialize their set of primal and dual variables with the same dimensions across all agents. This requires that each agent must know the total number of agents $n$. Since this may not be possible, the agents can have a common upper-estimate of $n$ and initialize their variables as per this estimate.
    
    \noindent Each agent $i$ also has to know \emph{which agents} are its neighbours to incorporate the constraint $(W_i)_{ij} = 0, (i,j) \notin E$ in its primal update step \eqref{eq:wi_final}. For this, all the agents should be indexed a priori and each agent must know its index. Then, in the first iteration of the algorithm, each agent can communicate with its neighbours to know their indices.
    
\end{rem}

Next, we present some numerical examples to illustrate the convergence of Algorithm \ref{alg:admm_parallel} and to compare our method with the centrally computed $W^*$ and the Metropolis weight matrix $\wm$.

\section{Examples}
\label{sec:examples}
\subsection{Fixed network}
\label{subsec:ex_static}

\begin{figure}[ht]
    \centering
    \includegraphics[scale=0.30]{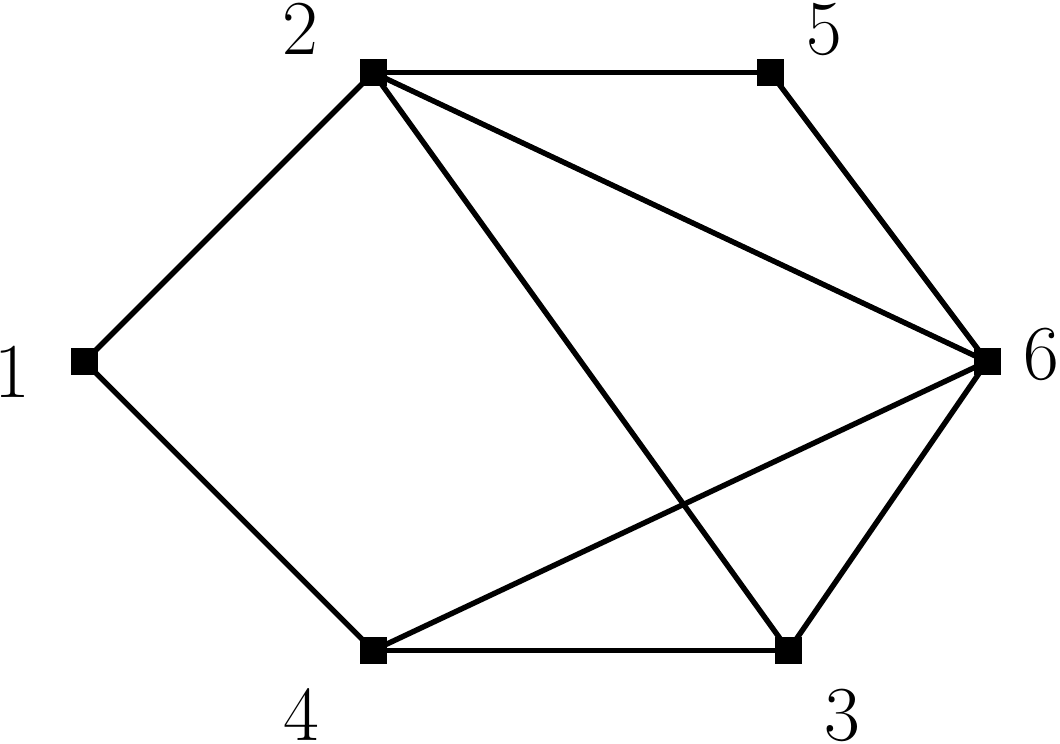}
    \caption{A connected network of $n = 6$ agents}
    \label{fig:network}
\end{figure}

\begin{table*}[ht]
\centering
\begin{tabular}{|L{2cm}|L{2cm}|L{2cm}|L{2cm}|} \hline
    & Centrally computed & \multicolumn{2}{c|}{Locally computed} \\ \cline{2-4}
    & $W^*$ & $\wm$ & $\wadmm$ \\
    & (centralized) & (Metropolis) & (Algorithm \ref{alg:admm_parallel}) \\ \hline
    $||W - \one \one^T/n||$ & $0.4492$ & $0.6724$ & $0.4519$ \\ \hline
\end{tabular}
\caption{Convergence factors of different weight matrices for the network shown in Fig. \ref{fig:network}}
\label{tab:comparison}
\end{table*}

\begin{figure}[ht]
    \centering
    \includegraphics[width=0.5\textwidth]{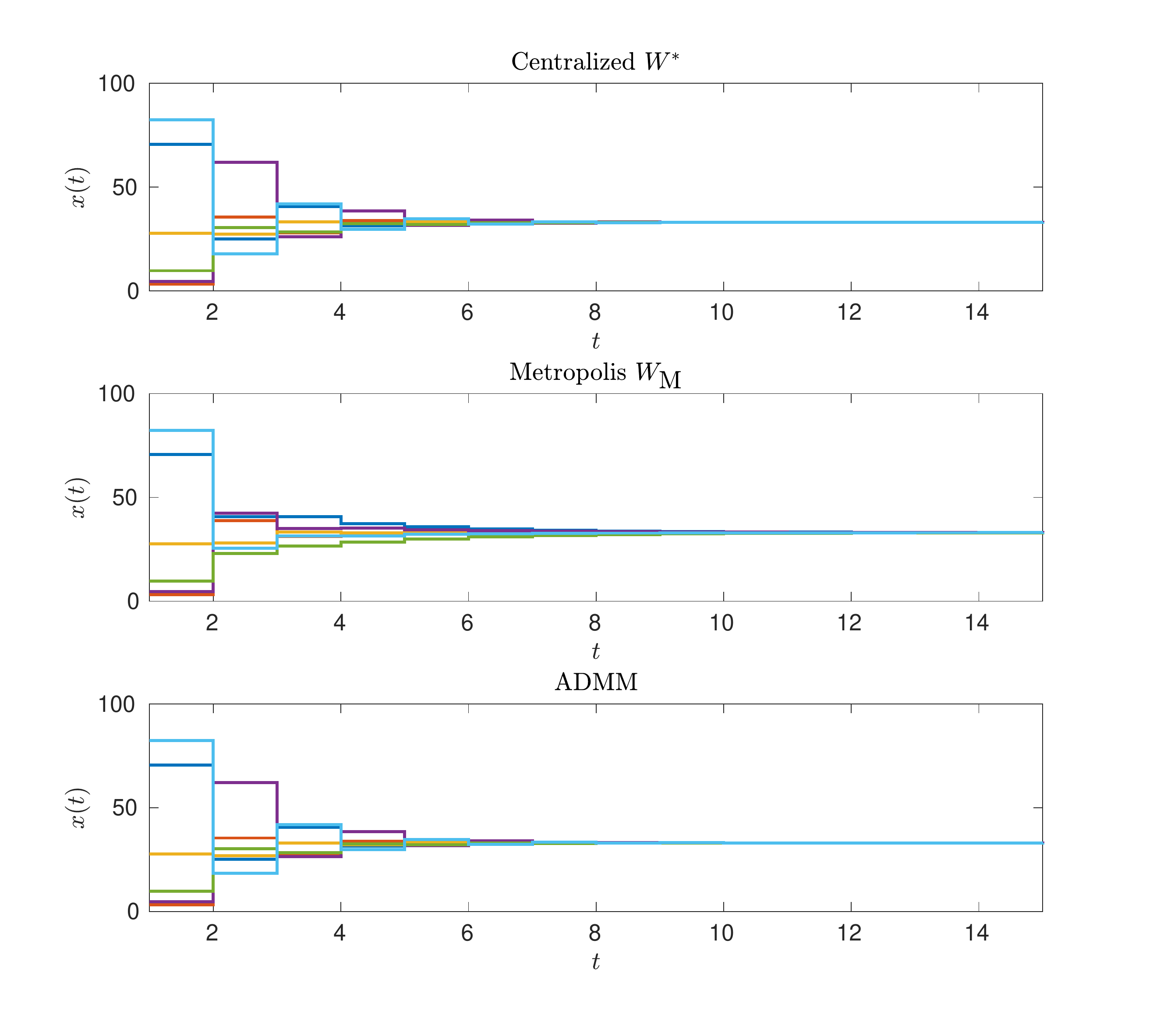}
    \caption{Convergence of the agent values to the average of the initial values with the three weight matrices $W^*$, $\wm$ and $\wadmm$ used in the consensus protocol for the network shown in Fig. \ref{fig:network}}
    \label{fig:states}
\end{figure}

\begin{figure}[ht]
    \centering
    \includegraphics[width=0.5\textwidth]{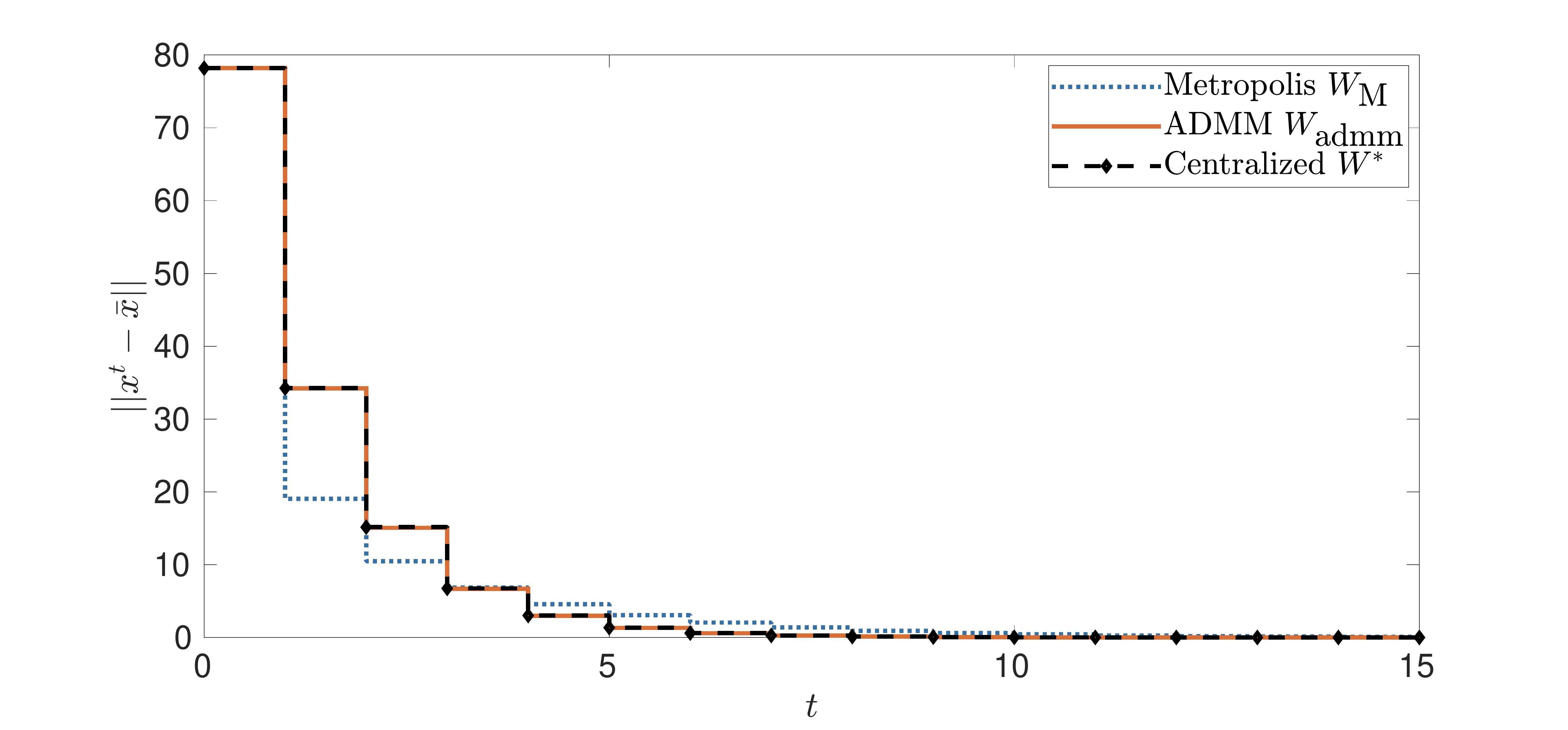}
    \caption{Convergence to zero of the consensus error with the three weight matrices $W^*$, $\wm$ and $\wadmm$ used in the consensus protocol for the network shown in Fig. \ref{fig:network}. It was observed that the error went below $0.1$ at $t = 9$ for the centralized and ADMM methods, while the Metropolis weights required $t = 14$.}
    \label{fig:error}
\end{figure}

\begin{figure}[ht]
    \psfrag{11^T}{\mathbb{11}^T}
    \centering
    \includegraphics[width=0.5\textwidth]{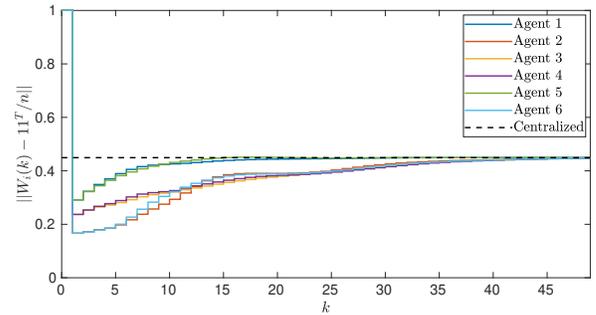}
    \caption{Convergence of the objective functions maintained by the agents to the centralized optimal value $||W^* - \one \one^T/n||$ as Algorithm \ref{alg:admm_parallel} progresses for the network shown in Fig. \ref{fig:network}}
    \label{fig:obj_fns}
\end{figure}

\begin{figure}[ht]
    \centering
    \includegraphics[width=0.5\textwidth]{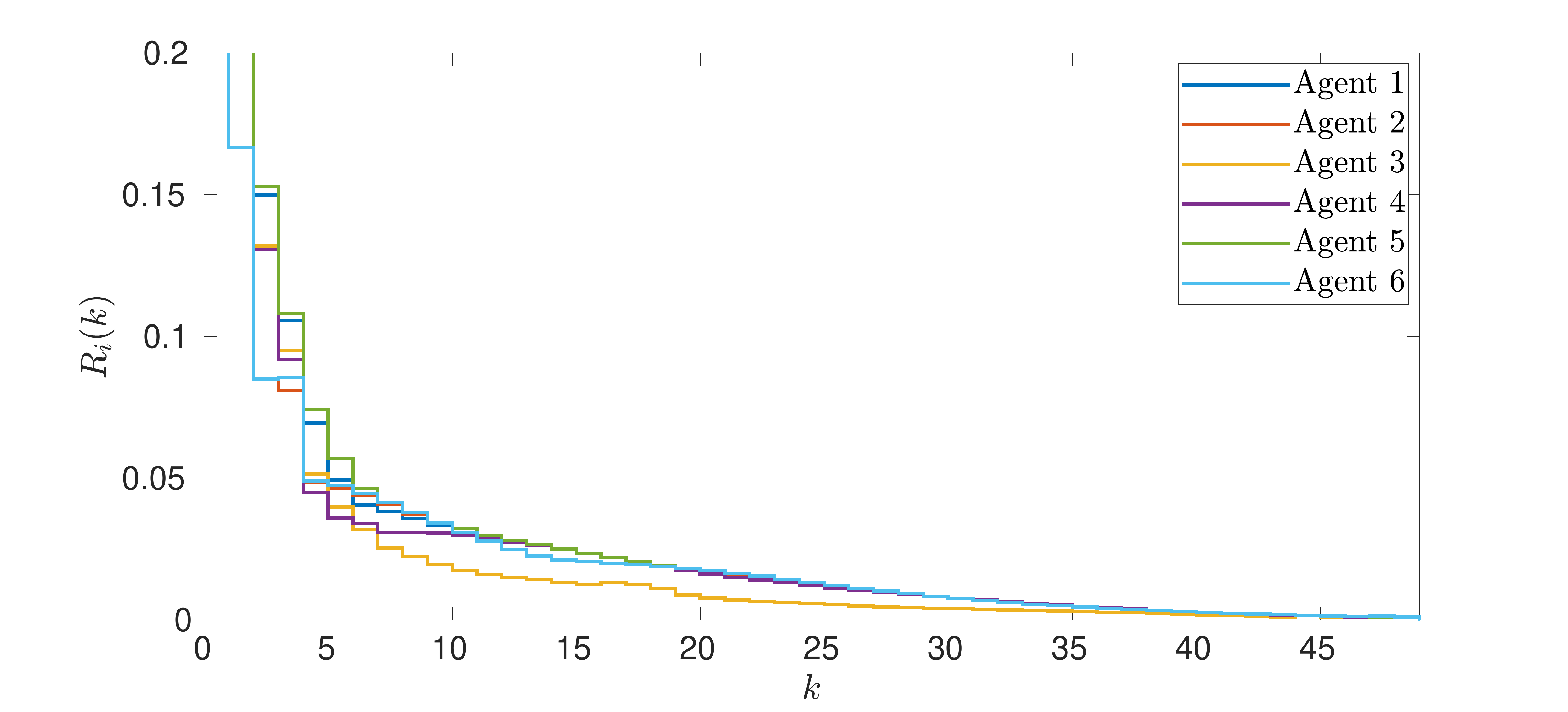}
    \caption{Convergence to zero of the maximum residuals maintained by the agents as Algorithm \ref{alg:admm_parallel} progresses for the network shown in Fig. \ref{fig:network}. The residuals become smaller than $\epsilon = 0.001$ (i.e., the stopping criterion is satisfied) at $k = 48$.}
    \label{fig:residuals}
\end{figure}

We consider an undirected, connected network of $n = 6$ agents as shown in Fig. \ref{fig:network}. For this network, we solve\footnote{We solve all optimization problems using the MOSEK solver \cite{mosek} in YALMIP toolbox \cite{Lofberg2004} on MATLAB \cite{matlab}.} the centralized problem \eqref{eq:W_optimization_2-norm} to get the optimal weight matrix $W^*$. Then, we compute the Metropolis weight matrix $\wm$ using \eqref{eq:W_metro}. 

Now, to calculate the weight matrix using our method, we run Algorithm \ref{alg:admm_parallel} in parallel on all $6$ agents using $\rho = 1/16$ as the penalty parameter\footnote{This value of $\rho$ was observed to give faster convergence of Algorithm \ref{alg:admm_parallel} for our examples.}, $\epsilon = 0.001$ for the stopping criteria. It was observed that the stopping criterion of Algorithm \ref{alg:admm_parallel} was satisfied at $k = 48$ (refer Remark \ref{rem:stop_crit} for details on how to check the stopping criterion). We want to find the convergence factor of the consensus protocol \eqref{eq:algo_vector} obtained using our method. For this, we define
\begin{align*}
    \wadmm = \begin{bmatrix} w_1^T \\ \vdots \\ w_n^T \end{bmatrix},
\end{align*}
where $w_i \in \R^{n \times 1}$ is the $i^\text{th}$ row of $W_i(48)$. The convergence factors of $\wadmm$ and other weight matrices are shown in Table \ref{tab:comparison}. It can be seen that the convergence factor of $\wadmm$ is quite close to that of the centrally computed $W^*$ and is much better than that of $\wm$. Now, we run the consensus protocol using these different weight matrices. Consider the randomly generated initial agent values given by $x(0) = \begin{bmatrix} 70.6046 &3.1833 &27.6923 &4.6171 &9.7132 &82.3458 \end{bmatrix}^T$. Their average is $x_\textrm{avg}(0) = 33.0260$. Fig. \ref{fig:states} shows that the agents reach consensus to this average value using each of the three weight matrices. Fig. \ref{fig:error} shows that for the centralized and ADMM methods, the consensus error $e(t) = x(t) - \bar{x}$ decays at almost the same rate, while the decay is slower for the Metropolis weights.

Next, we illustrate the convergence of Algorithm  \ref{alg:admm_parallel} used in calculating $\wadmm$. Fig. \ref{fig:obj_fns} shows the convergence of the objective functions maintained by each agent to the centralized optimal value $||W^* - \one \one^T/n|| = 0.4492$. This illustrates Statement \ref{thm_point:obj_convergence} of Theorem \ref{thm:main}. Fig. \ref{fig:residuals} shows the convergence to zero of the maximum residual at each agent as defined in \eqref{eq:max_residual}. Intuitively, the maximum residual $R_i(k)$ of agent $i$ captures how far is the objective function value $||W_i(k) - \one \one^T/n||$ maintained by the agent from the optimal objective value $||W^*-\one \one^T/n||$ (refer Appendix \ref{sec:appendix_stop_crit} for details on the maximum residual). Comparing Fig. \ref{fig:residuals} with Fig. \ref{fig:obj_fns}, we can observe that the objective functions are close to the optimal value when the residuals are close to zero. This asserts that our choice of the stopping criterion as derived in Appendix \ref{sec:appendix_stop_crit} is a good one.

\subsection{New agents enter a network: ADMM live --- a variation of Algorithm \ref{alg:admm_parallel}}
\label{subsec:admm_live_example}

\begin{figure}[ht]
    \centering
    \includegraphics[scale=0.30]{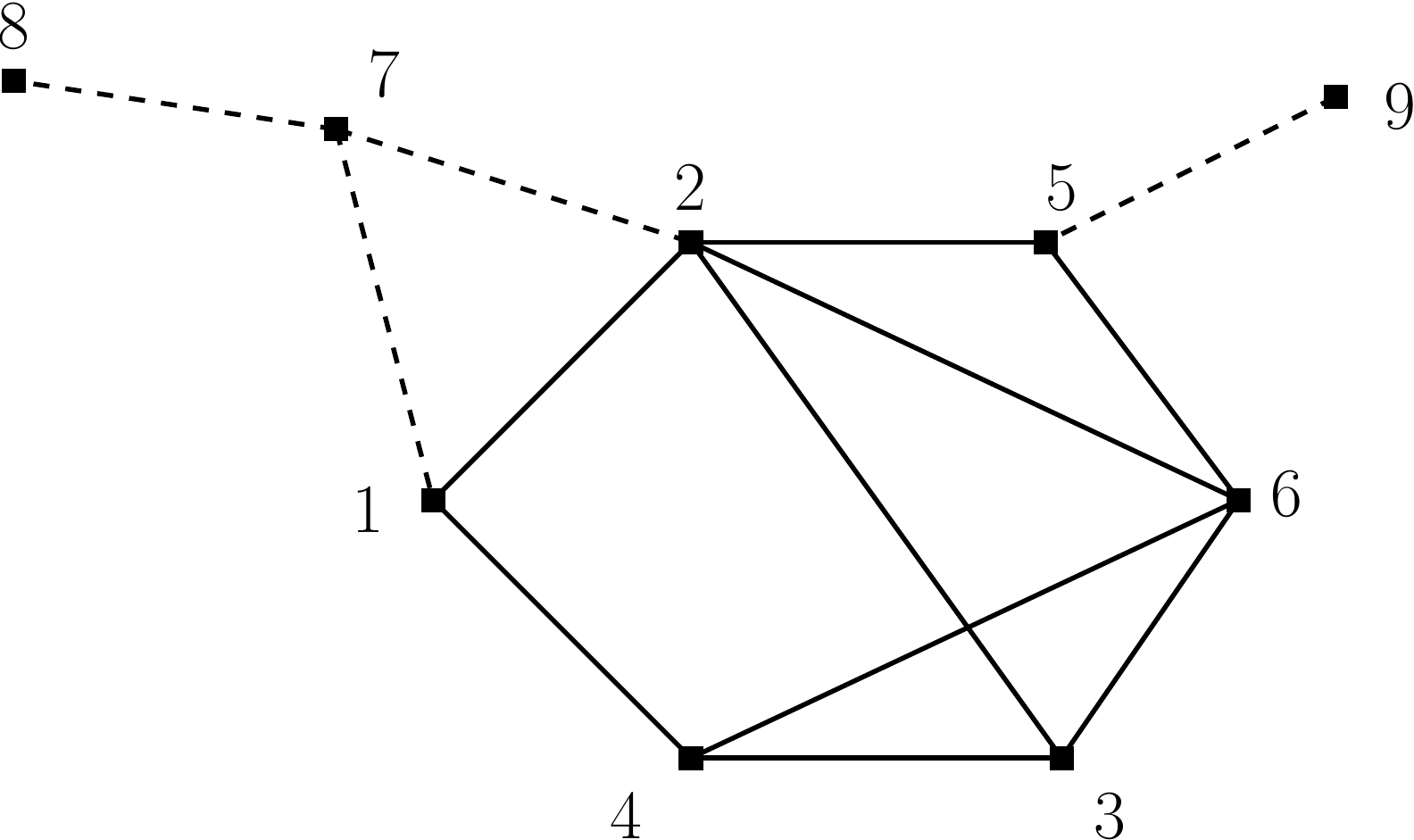}
    \caption{At $t = 0$, agents $1$ to $6$ are connected by a network. At $t = 30$, agents $7$, $8$ and $9$ enter the network.}
    \label{fig:network_new_agents}
\end{figure}

\begin{figure}[ht]
    \centering
    \includegraphics[width=0.5\textwidth]{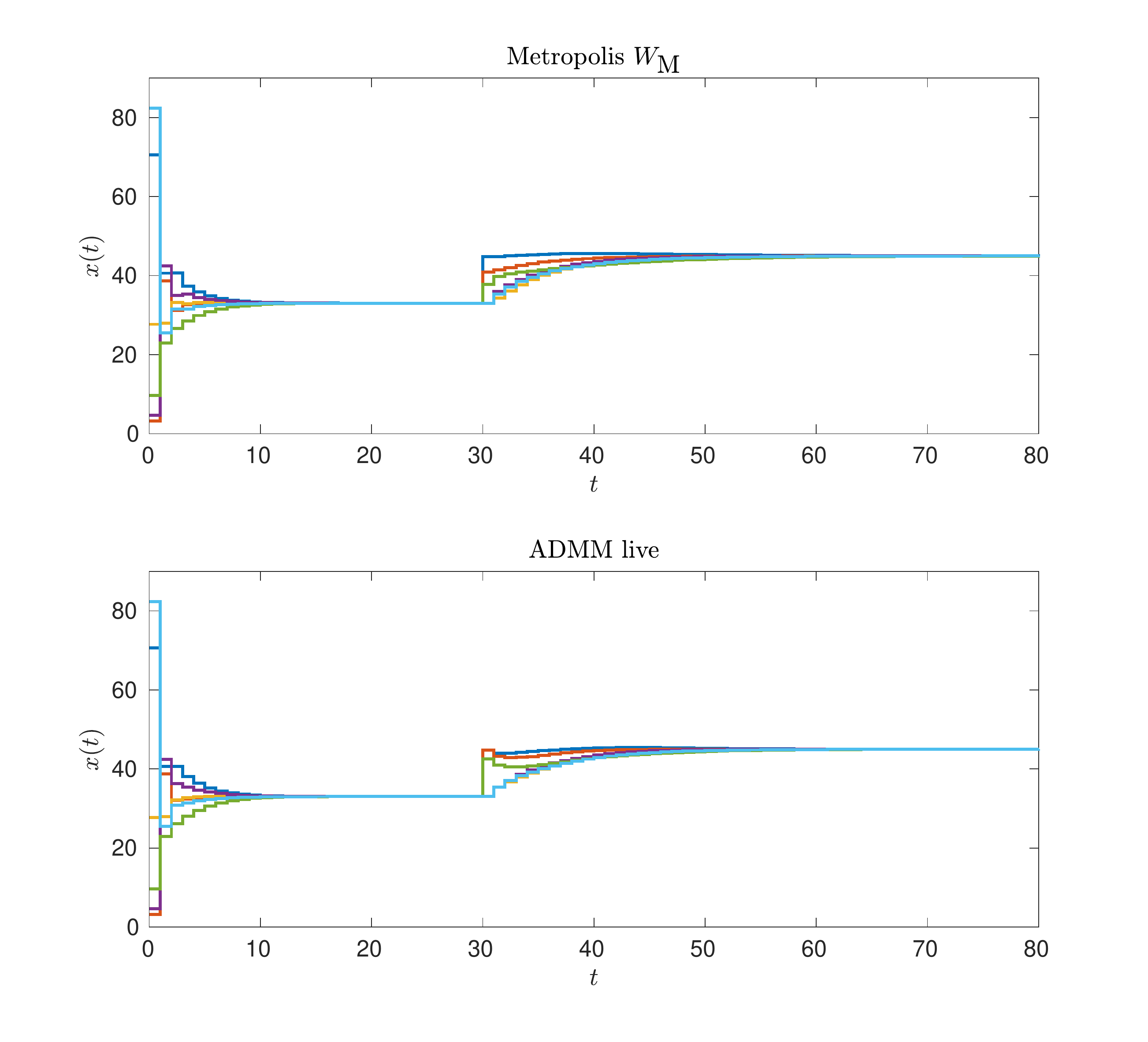}
    \caption{Evolution of the agent values with $\wm$ and ADMM live for the network shown in Fig. \ref{fig:network_new_agents}. In both methods, the agents reach consensus to the new average value after new agents enter the network at $t = 30$.}
    \label{fig:states_new_agents}
\end{figure}

\begin{figure}[ht]
    \centering
    \includegraphics[width=0.5\textwidth]{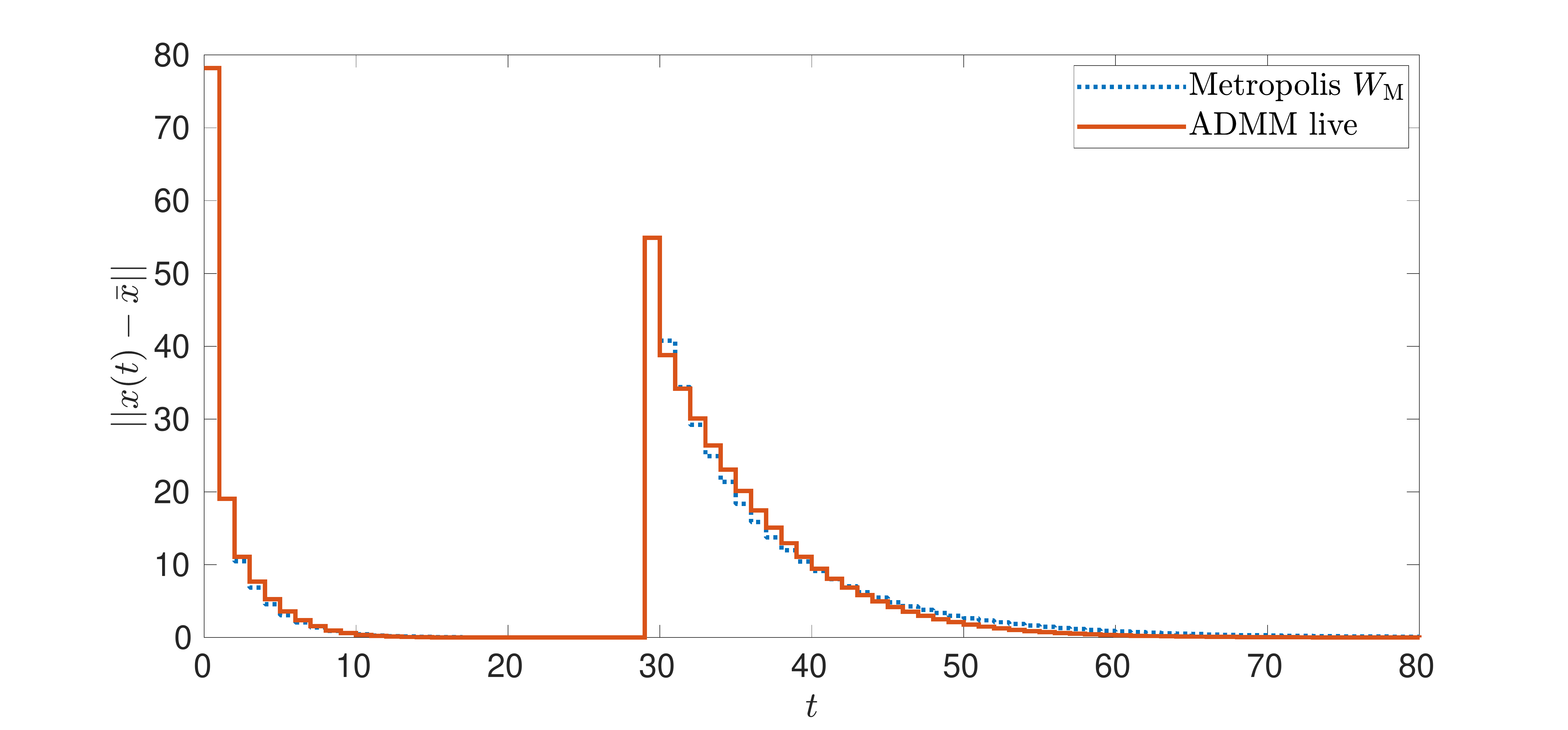}
    \caption{Evolution of the consensus error with $\wm$ and ADMM live for the network shown in Fig. \ref{fig:network_new_agents}. It was observed that, after the new agents entered the network at $t = 30$, the error went below $0.1$ at $t = 67$ for ADMM live, while the Metropolis weights required $t = 80$.}
    \label{fig:error_new_agents}
\end{figure}

\begin{figure}[ht]
    \centering
    \includegraphics[width=0.5\textwidth]{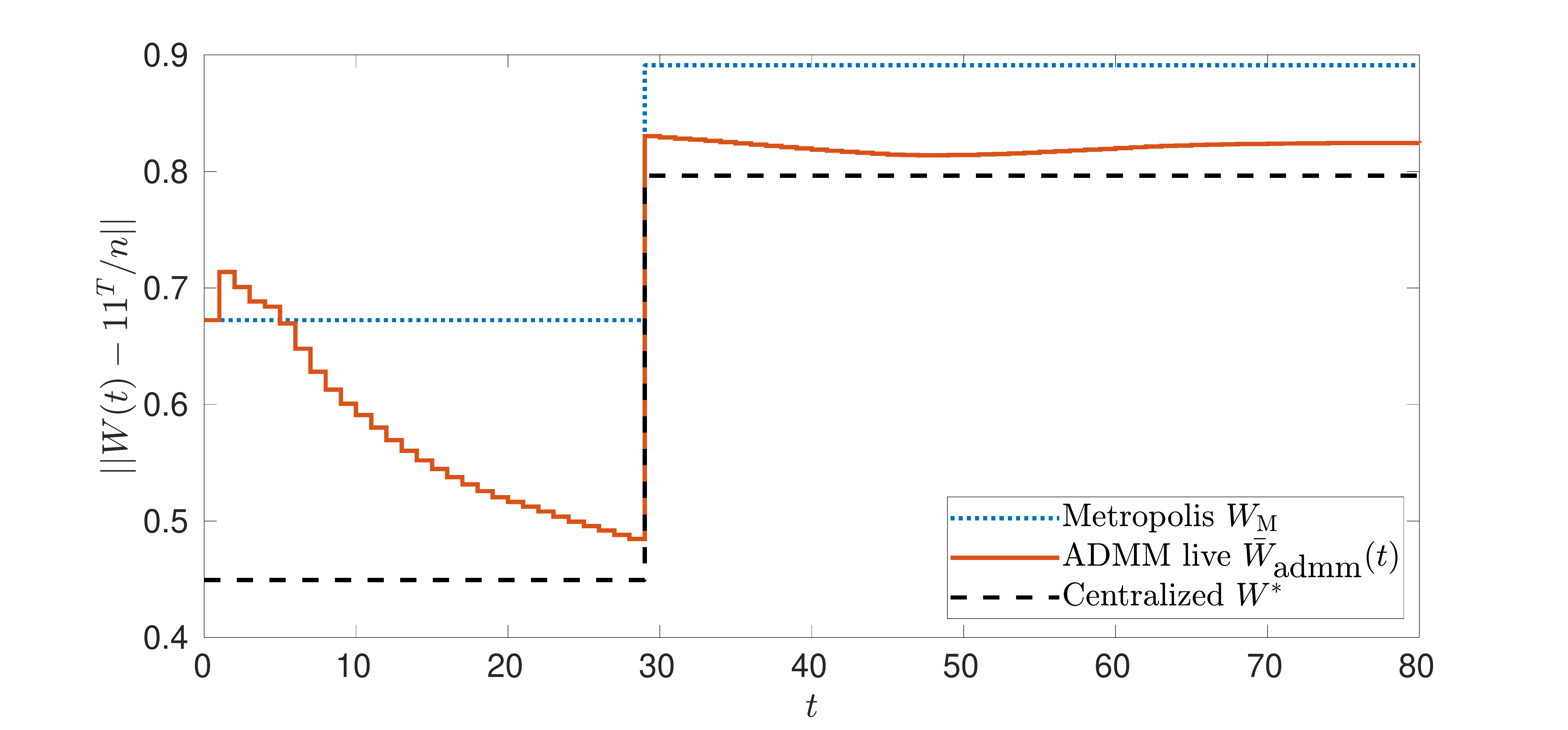}
    \caption{Convergence factor of the consensus protocol with $\wm$ and ADMM live for the network shown in Fig. \ref{fig:network_new_agents}}
    \label{fig:per_step_new_agents}
\end{figure}

In the previous example, we compared the convergence factor of $\wadmm$ with the centralized $W^*$ and the Metropolis $\wm$ and observed that the performance of $\wadmm$ is better than $\wm$. However, Algorithm \ref{alg:admm_parallel} requires some iterations of communication and computation in finding $\wadmm$ in a distributed manner. While certain applications do have some initial buffer time to compute the weights, few applications might not have this. For such cases we propose the following variation where at each iteration of Algorithm \ref{alg:admm_parallel}, we employ the consensus protocol \eqref{eq:algo_vector} using the \emph{latest available} weights.
More specifically, for each iteration $k \geq 0$, define 
\begin{align*}
    \wtadmm(k) = \begin{bmatrix} w_1(k)^T \\ \vdots \\ w_n(k)^T \end{bmatrix},
\end{align*}
where $w_i(k) \in \R^{n \times 1}$ is the $i^\text{th}$ row of $W_i(k)$. Note that for small values of $k$, the matrix $\wtadmm(k)$ may not satisfy the constraint 
\begin{align}
    \label{eq:wtadmm_consitions}
    \wtadmm(k) \one = \one, \ \wtadmm(k)^T \one = \one.
\end{align}
However, it is necessary that a weight matrix satisfies this constraint at each instant of time for the agents to achieve average consensus using the consensus protocol (refer Proposition \ref{prop:boyd_W_condition}). Hence, we modify $\wtadmm(k)$ as follows. Define a new weight matrix $\wbadmm(k)$ as
\begin{align*}
    \wbadmm(k)_{ij} = \begin{cases}
        \min\{\wtadmm(k)_{ij},\wtadmm(k)_{ji}\}, & i \neq j, \\
        1 - \sum_{l \neq i, l = 1}^n \wbadmm(k)_{il}, & i = j,
    \end{cases}
\end{align*}
i.e., at each iteration $k$, we convert $\wtadmm(k)$ into a symmetric matrix by replacing its $(i,j)^\text{th}$ and $(j,i)^\text{th}$ elements with the smaller one among the two. In context of Algorithm \ref{alg:admm_parallel}, this can be done by exchanging weights among the neighbours. Then, each agent $i$ adjusts its self-weight such that the row (and hence the column) sums of $\wbadmm(k)$ are $1$. This technique is same as the one used in calculating the Metropolis weight matrix $\wm$ in \eqref{eq:W_metro}.  Now, $\wbadmm(k)$ satisfies the constraint \eqref{eq:wtadmm_consitions}. Further, to reduce the overall time required to reach consensus, we initialize $W_i(k)$ using the Metropolis weights, i.e., for each $i \in \{1, \dots, n\}$, set the $i^\text{th}$ row of $W_i(0)$ to be the same as that of $\wm$, with the rest of the rows being zero. This way, the agent values in the consensus protocol will not remain `idle' for $t = 0$ due to all weights being initialized at zero. Now, we update the agent values using $\wbadmm(k)$ in the consensus protocol, i.e., for all $t \geq 0$,
\begin{align}
    \label{eq:admm_protocol}
    x(t+1) = \wbadmm(t) x(t).
\end{align}
We call the above method of implementing Algorithm \ref{alg:admm_parallel} and consensus protocol \eqref{eq:admm_protocol} simultaneously as \emph{ADMM live}. Now, we compare the performance of ADMM live with the Metropolis weights using the following example. 

At time $t = 0$, suppose there are $6$ agents connected by the same network as shown in Fig. \ref{fig:network}. Now, at $t = 30$, three new agents enter the network as shown in Fig. \ref{fig:network_new_agents}. Suppose these new agents have `initial values' $x_7(30) = 80.0559, x_8(30) = 74.5847, x_9(30) = 52.1186$. Now, the new average value of the $9$ agents is $x_\textrm{avg}(0) = 44.9906$. Then, all the agents must compute a set of weights which involve these new agents and reach consensus to the new average value. The evolution of the agent values using the Metropolis $\wm$ and ADMM live is shown in Fig. \ref{fig:states_new_agents}. Note that the centralized method is not capable of dynamically computing new weights after a change in the network topology. Hence, we compare ADMM live with the Metropolis weights, both of which can handle a change in network topology since they compute weights locally. In Fig. \ref{fig:states_new_agents}, the agents try to reach consensus initially, until new agents enter the network at $t = 30$. At this point, all agents compute a new set of weights and reach the new average value asymptotically. Fig. \ref{fig:error_new_agents} shows the consensus error for the two methods. It is observed that since the convergence factor of the weights obtained by ADMM live converges to the optimal value, the error goes to zero faster than with the Metropolis $\wm$. In Fig. \ref{fig:per_step_new_agents}, we compare this changing convergence factor with the that of the fixed Metropolis weights. We can see that the convergence factor given by ADMM live is smaller than the Metropolis weights for most of the time. In fact, it can be observed from the figure that just after a few steps (at $k = 6$), the convergence factor of ADMM live starts outperforming that of the Metropolis weights. Also, this value approaches the optimal value given by the centralized optimal solution $W^*$.
Hence we observe that the performance of ADMM live is better than that of Metropolis.

\subsection{Average computation time of Algorithm \ref{alg:admm_parallel}}

\begin{figure}[ht]
    \centering
    \includegraphics[width=0.5\textwidth]{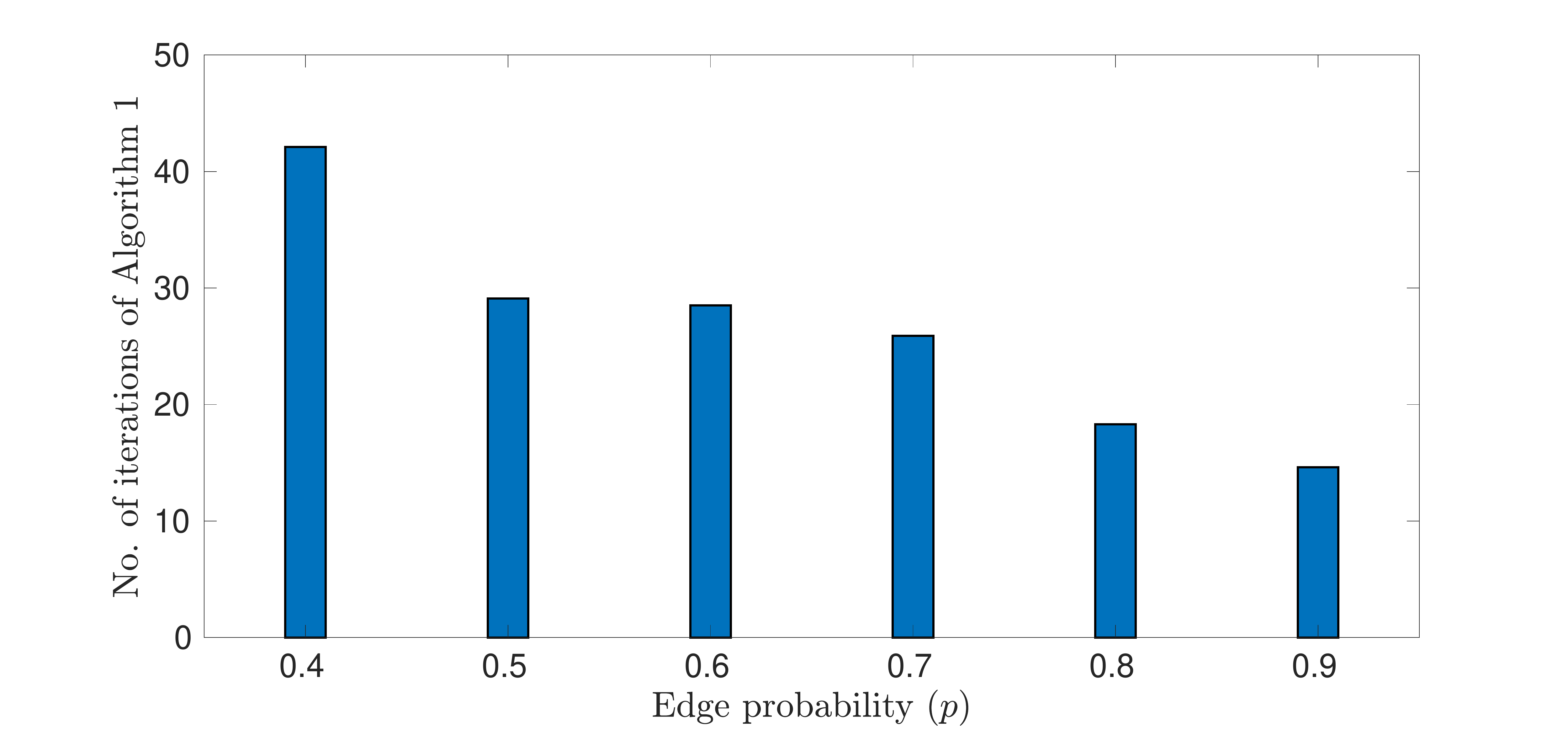}
    \caption{Number of iterations required for Algorithm \ref{alg:admm_parallel}, averaged over $10$ Erdos-Renyi random graphs. Each graph has $n = 10$ nodes and an edge probability of $p$.}
    \label{fig:avg_no_of_itr}
\end{figure}

We analyze what is the average number of iterations required for Algorithm \ref{alg:admm_parallel} to satisfy its stopping criterion. For this, we construct a set of Erdos-Renyi (ER) random graphs as follows. Consider a network of $n = 10$ agents. For any pair of nodes, we place an edge between them with a probability $p$. For each fixed $p \in \{0.4,0.5,0.6,0.7,0.8,0.9\}$, we generate $10$ ER random graphs. Then, we run Algorithm \ref{alg:admm_parallel} on each of them with $\epsilon = 0.01$ for the stopping criterion. For each $p$, we take the average of the number of iterations required for the algorithm to stop, averaged over all $10$ random graphs. The result is shown in Fig. \ref{fig:avg_no_of_itr}. It can be seen that for denser graphs (larger $p$), the algorithm converges faster. 

Next, we give some concluding remarks.

\section{Conclusion}
\label{sec:conclusion}

We proposed a method for agents in a network to locally compute a set of weights for the discrete-time distributed average consensus protocol. The weights are computed through an iterative algorithm derived from ADMM. The algorithm requires each agent to maintain and update a set of variables by solving an optimization problem and by exchanging these variables with its neighbours. We showed that the convergence factor of the consensus protocol with these locally computed weights converges to the optimal value given by the centralized solution. Through a numerical example, we showed that these weights perform better than the locally computed Metropolis weights in terms of the convergence factor. 

\begin{table*}[ht]
\noindent\rule{\textwidth}{1pt}
\normalsize{
Augmented Lagrangian:
\begin{align}
    \label{eq:aug_lag}
    L_\rho & = \frac{1}{n} \sum_{i=1}^n ||W_i - \one \one^T / n|| + \sum_{i=1}^n a_i^T(W_i \one - \one) + \sum_{i=1}^n b_i^T(W_i^T \one - \one) + \sum_{(i,j) \in E} \big[\tr((W_i - X_{ij})^T C_{ij}) \nonumber \\
    & \qquad + \tr((W_j - X_{ji})^T D_{ij})\big] + \frac{\rho}{2} \bigg[ \sum_{i=1}^n \big(||W_i \one - \one||^2 + ||W_i^T \one - \one||^2\big) + \sum_{(i,j) \in E} \big(||W_i - X_{ij}||_F^2 + ||W_j - X_{ji}||_F^2\big) \bigg]
\end{align}}
\end{table*}

\begin{table*}[ht]
\noindent\rule{\textwidth}{1pt}
\normalsize{
Update rule in Algorithm \ref{alg:admm_serial}:
\begin{align}
    \label{eq:wi_initial}
    W_i(k+1) & = \underset{W_i:(W_i)_{ij} = 0, (i,j) \notin E}{\argmin} \bigg\{\frac{1}{n}||W_i - \one \one^T / n|| + a_i(k)^T(W_i \one - \one) + b_i(k)^T(W_i^T \one - \one) + \sum_{j \in N_i} \tr(W_i^T C_{ij}(k)) \nonumber \\
    & \qquad + \frac{\rho}{2} \bigg[ ||W_i \one - \one||^2 + ||W_i^T \one - \one||^2 + \sum_{j \in N_i}||W_i - X_{ij}(k)||_F^2 \bigg] \bigg\}
\end{align}}
\end{table*}

\begin{table*}[ht]
\noindent\rule{\textwidth}{1pt}
\normalsize{
Update rule in Algorithm \ref{alg:admm_parallel}:
\begin{align}
    \label{eq:wi_final}
    W_i(k+1) & = \underset{W_i:(W_i)_{ij} = 0, (i,j) \notin E}{\argmin} \bigg\{\frac{1}{n}||W_i - \one \one^T / n|| + a_i(k)^T(W_i \one - \one) + b_i(k)^T(W_i^T \one - \one) + \tr(W_i^T M_i(k)) \nonumber \\
    & \qquad + \frac{\rho}{2} \bigg[ ||W_i \one - \one||^2 + ||W_i^T \one - \one||^2 + \sum_{j \in N_i}\bigg|\bigg|W_i - \frac{W_i(k)+W_j(k)}{2}\bigg|\bigg|_F^2 \bigg] \bigg\},
\end{align}
where $M_i(k) = \sum_{j \in N_i} C_{ij}(k)$.} 

\noindent\rule{\textwidth}{1pt}
\end{table*}

\appendices

\section{Proof of Theorem \ref{thm:main}}
\label{sec:appendix_proof_main}

We have to show that Algorithm \ref{alg:admm_parallel} is an ADMM implementation of problem \eqref{eq:W_i_optimization}.
\begin{rem}
    \label{rem:admm_serial}
    Note that if we try to solve problem \eqref{eq:W_i_optimization} directly using ADMM, although the ADMM algorithm can run in a distributed manner, the $W_i(k+1)$ updates will have to be performed `serially' across different agents, i.e., agent $i$ would have to wait for agents $1, \dots, i-1$ to update their values before it can perform its update. This is because the $W_i$'s corresponding to different agents appear in the same constraint $W_i = W_j, (i,j) \in E$ in problem \eqref{eq:W_i_optimization}. 
\end{rem}
To overcome the issue stated in Remark \ref{rem:admm_serial} and to enable `parallel' implementation of ADMM, we introduce dummy variables\footnote{This trick is taken from \cite[Section~3.4]{10.5555/59912}.} $X_{ij}, X_{ji}$ for each $(i,j) \in E$ in \eqref{eq:W_i_optimization}. Now, we get
\begin{align}
    \label{eq:optimization_xij's}
    & \mathrm{minimize} & & \sum_{i=1}^n\frac{||W_i - \one \one^T/n||}{n} \tag{P4} \\
    & \mathrm{subject \ to} & & W_i \one = \one, \ W_i^T \one = \one, \ i \in \{1, \dots, n\}, \nonumber \\
    & & & (W_i)_{ij} = 0, (i,j) \notin E, i \in \{1, \dots, n\}, \nonumber \\
    & & & W_i = X_{ij}, \ W_j = X_{ji}, \ X_{ij} = X_{ji}, \ (i,j) \in E. \nonumber 
\end{align}
Thus, the $W_i$'s corresponding to different agents have now been `decoupled'. Above problem is equivalent to \eqref{eq:W_i_optimization}. To solve this problem using ADMM, we define the augmented Lagrangian for the problem as given in \eqref{eq:aug_lag}. Here, $\rho > 0$ is the penalty parameter and $a_i,b_i,C_{ij},D_{ij}$ are the dual variables for the respective constraints. Algorithm \ref{alg:admm_serial} is the standard ADMM algorithm for \eqref{eq:optimization_xij's}. Note that the constraints $X_{ij} = X_{ji}, (i,j) \in E$ and $(W_i)_{ij} = 0, (i,j) \notin E$ have not been dualized but are incorporated in the primal update steps of the algorithm.

\begin{algorithm}
\caption{ADMM implementation of problem \eqref{eq:W_i_optimization}}
\label{alg:admm_serial}
\begin{algorithmic}
	\STATE \textbf{Initialize:} $\rho > 0, W_i(0) = 0_{n \times n}, X_{ij}(0) = 0_{n \times n}, X_{ji}(0) = 0_{n \times n}, a_i(0) = 0_{n \times 1}, b_i(0) = 0_{n \times 1}, C_{ij}(0) = 0_{n \times n}, D_{ij}(0) = 0_{n \times n}$
	\STATE \textbf{Iterate:} 
	\FOR{$k \geq 0$}
	\STATE \textbf{Primal update:}
	\begin{align*}
	    W_i(k+1) &= \underset{W_i:(W_i)_{ij} = 0, (i,j) \notin E}{\argmin} L_\rho, i \in \{1, \dots, n\} \\
        X_{ij}(k+1) &= \underset{X_{ij}: X_{ij} = X_{ji}}{\argmin} L_\rho,  (i,j) \in E \\
        X_{ji}(k+1) &= \underset{X_{ji}: X_{ij} = X_{ji}}{\argmin} L_\rho, (i,j) \in E
	\end{align*}
	\STATE \textbf{Dual update:} \begin{align}
	    a_i(k+1) &= a_i(k) + \rho(W_i(k+1) \one - \one), i \in \{1, \dots, n\} \nonumber \\
	    b_i(k+1) &= b_i(k) + \rho(W_i(k+1)^T \one - \one), i \in \{1, \dots, n\} \nonumber \\
	    \label{eq:cij_original} 
	    C_{ij}(k+1) &= C_{ij}(k) \nonumber \\
	    & \quad + \rho(W_i(k+1) - X_{ij}(k+1)), (i,j) \in E \\
	    \label{eq:dij_original}
	    D_{ij}(k+1) &= D_{ij}(k) \nonumber \\
	    & \quad + \rho(W_j(k+1) - X_{ji}(k+1)), (i,j) \in E
	\end{align}
	\ENDFOR
\end{algorithmic}
\end{algorithm}

Algorithm \ref{alg:admm_serial} in its original form cannot be run in parallel. One of the reasons for this is that for any agent $i$, the $W_i(k+1)$ update is the $\argmin$ of $L_\rho$ evaluated at $W_1 = W_1(k+1),\dots, W_{i-1} = W_{i-1}(k+1)$. Thus, agent $i$ has to wait for the `previous agents' to perform their updates and hence the issue mentioned in Remark \ref{rem:admm_serial} still exists. However, due to our introduction of $X_{ij}, X_{ji}$'s, with some algebraic manipulations, we can show that Algorithm \ref{alg:admm_parallel} and Algorithm \ref{alg:admm_serial} are equivalent. The former can be executed in a parallel manner as argued in Remark \ref{rem:alg_distributed_manner}.

\begin{lem}
    Algorithm \ref{alg:admm_parallel} is equivalent to Algorithm \ref{alg:admm_serial}.
\end{lem}
\begin{proof}

We will simplify the primal and dual update steps in Algorithm \ref{alg:admm_serial} to arrive at Algorithm \ref{alg:admm_parallel}.

Consider an arbitrary $(i,j) \in E$. The primal update $X_{ij}(k+1)$ can be evaluated by equating the gradient of $L_\rho$ with respect to $X_{ij}$ to zero. We enforce the constraint $X_{ij} = X_{ji}$ when evaluating the gradient. This gives the following equation.
\begin{align}
    \label{eq:grad_wrt_xij}
    & -C_{ij}(k) - D_{ij}(k) + \rho [X_{ij}(k+1) - W_i(k+1) \nonumber \\
    & \qquad + X_{ij}(k+1) - W_j(k+1)] = 0
\end{align}
This implies
\begin{align}
    \label{eq:xij_with_cd}
    & X_{ij}(k+1) = X_{ji}(k+1) \nonumber \\
    & \qquad = \frac{C_{ij}(k) + D_{ij}(k)}{2 \rho} + \frac{W_i(k+1) + W_j(k+1)}{2}.
\end{align}
Substituting these in the dual updates \eqref{eq:cij_original}, \eqref{eq:dij_original}, we have
\begin{align}
    \label{eq:cij_initial}
    & C_{ij}(k+1) = C_{ij}(k) \nonumber \\
    & \qquad + \rho\bigg(\frac{W_i(k+1) - W_j(k+1)}{2} - \frac{C_{ij}(k) + D_{ij}(k)}{2 \rho}\bigg), \\
    \label{eq:dij_initial}
	& D_{ij}(k+1) = D_{ij}(k) \nonumber \\
	& \qquad + \rho\bigg(\frac{W_j(k+1) - W_i(k+1)}{2} - \frac{C_{ij}(k) + D_{ij}(k)}{2 \rho}\bigg).
\end{align}
Adding the above two equations, we get
\begin{align*}
    C_{ij}(k+1) + D_{ij}(k+1) = 0.
\end{align*}
Initializing $C_{ij}(0) = D_{ij}(0) = 0$ implies 
\begin{align}
    \label{eq:cij+dij=0}
    C_{ij}(k) + D_{ij}(k) = 0 
\end{align}
for all $k \geq 0$. Substituting this back in \eqref{eq:cij_initial}, \eqref{eq:dij_initial} gives
\begin{align}
    \label{eq:cij_final}
    C_{ij}(k+1) &= C_{ij}(k) + \rho\bigg(\frac{W_i(k+1) - W_j(k+1)}{2}\bigg), \\
    \label{eq:dij_final}
	D_{ij}(k+1) &= D_{ij}(k) +  \rho\bigg(\frac{W_j(k+1) - W_i(k+1)}{2}\bigg).
\end{align}
Further, substituting $C_{ij}(k) + D_{ij}(k) = 0$ in \eqref{eq:xij_with_cd} gives
\begin{align}
    \label{eq:xij}
    X_{ij}(k+1) = X_{ji}(k+1) = \frac{W_i(k+1) + W_j(k+1)}{2}.
\end{align}

Next, consider an arbitrary $i \in \{1, \dots, n\}$. Then, the update $W_i(k+1)$ in Algorithm \ref{alg:admm_serial} can be evaluated as given in \eqref{eq:wi_initial}. Here, in finding the $\argmin$ of $L_\rho$, we have ignored those terms in $L_\rho$ which are independent of $W_i$.

Now, in \eqref{eq:wi_initial}, we can substitute $X_{ij}(k)$ from \eqref{eq:xij}. Further, we can simplify the term $\sum_{j \in N_i} \tr(W_i^T C_{ij}(k))$ as follows. Define $M_i(k) = \sum_{j \in N_i} C_{ij}(k)$. Then, 
\begin{align*}
    M_i(k+1) &= \sum_{j \in N_i} C_{ij}(k+1) \nonumber \\
    & = \sum_{j \in N_i} \bigg[C_{ij}(k) + \rho\bigg(\frac{W_i(k+1) - W_j(k+1)}{2}\bigg)\bigg] \nonumber \\
    & = M_i(k) + \rho \sum_{j \in N_i} \bigg(\frac{W_i(k+1) - W_j(k+1)}{2}\bigg)
\end{align*}
and $\sum_{j \in N_i} \tr(W_i^T C_{ij}(k)) = \tr(W_i^T M_i(k))$. Thus, $W_i(k+1)$ can be computed as given in \eqref{eq:wi_final}.

In conclusion, going from Algorithm \ref{alg:admm_serial} to Algorithm \ref{alg:admm_parallel}, the primal variables $X_{ij}(k),X_{ji}(k)$ are no longer required to be maintained explicitly. Further, the dual variables $C_{ij}(k),D_{ij}(k)$ have been combined into $M_i(k)$. The other dual variables $a_i(k),b_i(k)$ remain unchanged. 

\end{proof}

\section{Stopping criterion for Algorithm \ref{alg:admm_parallel}}
\label{sec:appendix_stop_crit}

We derive a stopping criterion for Algorithm \ref{alg:admm_parallel}. A good stopping criterion for the ADMM of a constrained optimization problem is when the optimality conditions of the problem are satisfied with some tolerance \cite{10.1561/2200000016}. The optimality conditions for problem \eqref{eq:optimization_xij's} are that the gradient of the \textit{unaugmented} Lagrangian $L_0$ ($\rho = 0$ in \eqref{eq:aug_lag})
with respect to the primal variables $W_i$ and $X_{ij}$ must be zero and the constraints must be satisfied, i.e., for all $i \in \{1, \dots, n\},$
\begin{align}
    & \frac{\partial}{\partial W_i} L_0 = 0, \ \frac{\partial}{\partial X_{ij}} L_0 = 0, j \in N_i, \nonumber \\
    & W_i \one = \one, \ W_i^T \one = \one, \nonumber \\
    \label{eq:optimality_conditions}
    & (W_i)_{ij} = 0, j \notin N_i, \ W_i = W_j, j \in N_i.
\end{align}
We check under what conditions do the primal and dual variables in Algorithm \ref{alg:admm_serial} satisfy these conditions.

First, we show that the conditions $(\partial/\partial W_i) L_0 = 0$ and $(\partial/\partial X_{ij}) L_0 = 0$ are always satisfied by the primal and dual variable iterates $W_i(k+1),a_i(k+1),b_i(k+1),C_{ij}(k+1),D_{ij}(k+1)$ for all $k \geq 0$.

By definition of $W_i(k+1)$ as given in Algorithm \ref{alg:admm_serial}, we know that $W_i(k+1)$ minimizes $L_\rho$ given in \eqref{eq:aug_lag}. Hence, $(\partial/\partial W_i) L_\rho$ evaluated at $W_i = W_i(k+1)$ must be zero. This means 
\begin{align*}
    & \frac{\partial}{\partial W_i} \bigg(\frac{1}{n} ||W_i - \one \one^T / n||\bigg)\bigg\rvert_{W_i = W_i(k+1)} + a_i(k) \one^T \\
    & + \one b_i(k)^T + \sum_{j \in N_i} C_{ij}(k) + \rho \bigg[(W_i(k+1) \one - \one)\one^T \nonumber \\
    & + \one(W_i(k+1)^T \one - \one) + \sum_{j \in N_i} (W_i(k+1) - X_{ij}(k))\bigg] = 0.
\end{align*}
From the dual update equations $a_i(k+1),b_i(k+1),C_{ij}(k+1)$ as given in Algorithm \ref{alg:admm_serial}, we can combine some terms in the above equation to get
\begin{align*}
    & \frac{\partial}{\partial W_i} \bigg(\frac{1}{n} ||W_i - \one \one^T / n||\bigg)\bigg\rvert_{W_i = W_i(k+1)} + a_i(k+1) \one^T \\
    & + \one b_i(k+1)^T + \sum_{j \in N_i} C_{ij}(k+1) = 0.
\end{align*}
Above is precisely the condition $(\partial/\partial W_i) L_0 = 0$ evaluated at $W_i(k+1),a_i(k+1),b_i(k+1),C_{ij}(k+1)$. Thus, $(\partial/\partial W_i) L_0 = 0$ is always satisfied by the primal and dual variable iterates.

Now, consider the second condition $(\partial/\partial X_{ij}) L_0 = 0$. Evaluating the gradient of $L_0$, with the constraint that $X_{ij} = X_{ji}$, this condition can be written as $C_{ij} + D_{ij} = 0$. From \eqref{eq:cij+dij=0}, we know that this is true for $C_{ij}(k),D_{ij}(k)$ for all $k \geq 0$. Thus, $(\partial/\partial X_{ij}) L_0 = 0$ is always satisfied by the primal and dual iterates of the algorithm.

Thus, from the optimality conditions \eqref{eq:optimality_conditions}, only the constraint conditions are \emph{not} satisfied by the iterates of the variables in Algorithm \ref{alg:admm_serial}. Define the residuals of the constraints as
\begin{align}
    \label{eq:residuals}
    r^1_i(k) &= ||W_i(k) \one - \one||/\sqrt{n}, \nonumber \\
    r^2_i(k) &= ||W_i(k)^T \one - \one||/\sqrt{n}, \nonumber \\
    r^3_{ij}(k) &= ||W_i(k) - W_j(k)||_F/n, j \in N_i, \nonumber \\ 
    r^4_{ij}(k) &= |(W_i(k))_{ij}|, j \notin N_i,
\end{align}
for all $i \in \{1, \dots, n\}$. Here, we have divided by $n$ or $\sqrt{n}$ to compensate for the dimension of the quantity inside the norm. Now, we define the maximum residual at agent $i$ as
\begin{align}
    \label{eq:max_residual}
    R_i(k) = \max \{r^1_i(k), r^2_i(k), r^3_{ij}(k), r^4_{ij}(k): j \in N_i\}.
\end{align}
Then, the optimality conditions \eqref{eq:optimality_conditions} are satisfied with an $\epsilon$ tolerance for some $k$ if 
\begin{align}
    \label{eq:stopping_crit}
    R_i(k) \leq \epsilon \text{ for all } i \in \{1, \dots, n\}.
\end{align}
This is the stopping criterion for Algorithm \ref{alg:admm_parallel}. Remark \ref{rem:stop_crit} explains how this stopping criterion is used in the algorithm.

\section{Proof of Lemma \ref{lem:convex_relaxation}}
\label{sec:appendix_proof_conv_relax_lemma}

First, we recall the definition of a primitive matrix.
\begin{defn}
A non-negative matrix $W$ is said to be primitive with index $m \geq 1$ if $W^m > 0$.
\end{defn}
Primitive matrices are a special class of matrices which have all but one eigenvalue strictly within the circle described by the spectral radius of the matrix. We make use of this fact whose proof can be found in \cite[(8.3.16)]{meyer2000matrix}.
\begin{prop}{\cite{meyer2000matrix}}
\label{prop:primitive}
Suppose $W$ is a non-negative matrix which satisfies $W \one = \one, W^T \one = \one$. Then, $W$ is primitive if and only if $\rho(W - \one \one^T / n) < 1$.
\end{prop}

Given that the graph $G$ is connected, using the above result, we now construct a weight matrix $W$ which satisfies the consensus condition \eqref{eq:W_consensus_condition}. We will use this result to prove Lemma \ref{lem:convex_relaxation}.
\begin{lem}
    \label{lem:connected_G_weight_matrix}
    Given $G$ is a connected graph, consider a non-negative weight matrix $W$ which satisfies 
    \begin{align}
        \label{eq:primitive_W_for_conn_G}
        & W \one = \one, \ W^T \one = \one, \ W_{ij} > 0, (i,j) \in E.
    \end{align}
    Then, $\rho(W - \one \one^T/n) < 1$.
\end{lem}
\begin{proof}
We prove this result by showing that $W$ is a primitive matrix, i.e., we show that $\exists m \geq 1$ such that $W^m > 0$. 

For a given pair of nodes $(i,j)$, consider any $m_{ij} \geq 2$. Then, $$(W^{m_{ij}})_{ij} = \sum_{l_1=1}^n \sum_{l_2=1}^n \dots \sum_{l_{m_{ij}-1}=1}^n W_{il_1} W_{l_1l_2} \dots W_{l_{m_{ij}-1}j}.$$ 

Note that for $m_{ij} = 1$, $(W^{m_{ij}})_{ij} = W_{ij}$. Thus, $(W^{m_{ij}})_{ij}$ is the sum of the product of weights of all edges which are part of all paths of length at most $m_{ij}$ between nodes $(i,j)$ of the graph. Since the graph $G$ is connected, there is a path (of length at most $n-1$) between every pair of nodes $(i,j)$, i.e., for all $(i,j) \in \{1, \dots, n\}^2, \exists m_{ij} \in \{1, \dots, n-1\}$ such that $(W^{m_{ij}})_{ij} > 0$. Hence, we can conclude that $W^{n-1} > 0$, i.e., $W$ is primitive with an index of (at most) $n-1$.

Now, by Proposition \ref{prop:primitive}, we have $\rho(W - \one \one^T/n) < 1$.

\end{proof}

We are now ready to prove Lemma \ref{lem:convex_relaxation}. We have to show that any solution $W^*$ of \eqref{eq:W_optimization_2-norm} satisfies $\rho(W^* - \one \one^T/n) < 1$. Since $\rho(A) \leq ||A||$ for any matrix $A \in \R^{n \times n}$, it is sufficient to show that $||W^* - \one \one^T/n|| < 1$. Moreover, since $W^*$ is the optimal value of \eqref{eq:W_optimization_2-norm}, it is sufficient to show that, given $G$ is connected, $\exists W \in \R^{n \times n}$ which satisfies the constraints $$W \one = \one, \ W^T \one = \one, \ W_{ij} = 0, (i,j) \notin E$$ of \eqref{eq:W_optimization_2-norm} such that $||W - \one \one^T/n|| < 1$. We construct such a $W$ as follows. Consider a $W$ which satisfies the condition \eqref{eq:primitive_W_for_conn_G}. Then, we have
\begin{align*}
    ||W - \one \one^T/n||^2 & = \rho\big((W - \one \one^T/n)(W - \one \one^T/n)^T\big) \\
    & = \rho(WW^T - \one \one^T/n).
\end{align*}
It is easy to verify that $WW^T$ also satisfies the condition \eqref{eq:primitive_W_for_conn_G}. This implies, $\rho(WW^T - \one \one^T/n) < 1$, which implies $||W - \one \one^T/n|| < 1$. 

\bibliography{references}

\bibliographystyle{IEEEtran}

\begin{IEEEbiographynophoto}{Kiran Rokade} received the B.Tech. degree in electrical engineering from V.J.T.I., Mumbai, India, in 2016. and the an M.S. degree in electrical engineering from Indian Institute of Technology Madras, Chennai, India in 2020. Currently, he is pursuing his PhD at Cornell University. His research interests include multi-agent systems, optimization theory and game theory. 
\end{IEEEbiographynophoto}

\begin{IEEEbiographynophoto}{Rachel Kalpana Kalaimani}
received the B.E. degree in electrical and electronics engineering from P.S.G. College of Technology, Coimbatore, India, in 2009 and the Ph.D. degree in control from the Department of Electrical Engineering, Indian Institute of Technology Bombay, Mumbai, India, in 2014. After that she was a Post Doctoral Research Fellow at the Department of Electrical and Computer Engineering, University of Waterloo, Waterloo, ON, Canada. She is currently an Assistant Professor at Indian Institute of Technology Madras, Chennai, India. Her current research interests include distributed optimization, networked control systems and complex systems. Other interests include optimization of energy consumption in buildings, model predictive control, graph theoretic techniques, and numerical linear algebra.
\end{IEEEbiographynophoto}

\vfill

\end{document}